\documentclass[journal]{IEEEtran}
\ifCLASSINFOpdf
\usepackage[pdftex]{graphicx}

\else

\fi

\usepackage{bigints}

\usepackage{amssymb}   
\usepackage{amsxtra}
\usepackage{amscd}
\usepackage{amsthm}
 \usepackage{amscd}
 \usepackage{amsthm}
\usepackage{tikz}
\usetikzlibrary{matrix}
\usepackage{graphicx, subfigure}   
\usepackage{comment}
\usepackage{verbatim} 
\usepackage{dblfloatfix}    

\usepackage{float}

\usepackage{amsmath}
\usepackage{graphicx}

\usepackage{amssymb}

\usepackage{array}  
\usepackage{authblk}
\usepackage{multirow}  
\usepackage{arydshln}

\usepackage{color}

\usepackage{graphicx}
\usepackage{lipsum}
\usepackage{cite}
\newtheorem{lemma}{Lemma}

\renewcommand{\thefootnote}{\arabic{footnote}}
\usepackage{balance}
\usepackage[keeplastbox]{flushend}

  \author{Emre Arslan,~\IEEEmembership{Student Member,~IEEE,} Fatih Kilinc,~\IEEEmembership{Student Member,~IEEE,} Sultangali Arzykulov,~\IEEEmembership{Member,~IEEE,} \hspace*{25pt} Ali Tugberk Dogukan,~\IEEEmembership{Student Member,~IEEE,} Abdulkadir Celik,~\IEEEmembership{Senior Member,~IEEE,} \newline  Ertugrul Basar,~\IEEEmembership{Senior Member,~IEEE,} and Ahmad M. Eltawil,~\IEEEmembership{Senior Member,~IEEE}

\thanks{E. Arslan, F. Kilinc, A. T. Dogukan and E. Basar are with the Communications Research and Innovation Laboratory (CoreLab), Department of Electrical and Electronics Engineering, Ko\c{c} University, Sariyer 34450, Istanbul, Turkey. (e-mail: \{earslan18, fkilinc20, adogukan18, ebasar\}@ku.edu.tr).}
\thanks{S. Arzykulov, A. Celik and A. M. Eltawil are with Computer,
Electrical, and Mathematical Sciences \& Engineering (CEMSE) Division at King Abdullah University of Science and Technology (KAUST),
Thuwal, KSA 23955-6900 (e-mail: \{sultangali.arzykulov, abdulkadir.celik, ahmed.eltawil\}@kaust.edu.sa).}

}

\begin{document}
	\title{Reconfigurable Intelligent Surface Enabled Over-the-Air Uplink Non-orthogonal Multiple Access}
\maketitle	
	\begin{abstract}
Innovative reconfigurable intelligent surface (RIS) technologies are rising and recognized as
promising candidates to enhance 6G and beyond wireless communication systems. RISs acquire the ability to manipulate electromagnetic signals, thus, offering a degree of control over the wireless channel and the potential for many more benefits. Furthermore, active RIS
designs have recently been introduced to combat the critical double fading problem and other impairments passive RIS designs may possess. In this paper, the potential and flexibility
of active RIS technology are exploited for uplink systems to achieve virtual non-orthogonal multiple access (NOMA) through power disparity over-the-air rather than controlling transmit powers at the user side. Specifically, users with identical transmit
power, path loss, and distance can communicate with a base station sharing time and frequency resources in a NOMA fashion with the aid of the proposed hybrid RIS system. Here, the RIS is partitioned into active and passive parts and the distinctive partitions serve different users aligning their phases accordingly while introducing a power difference to the users’
signals to enable NOMA. First, the end-to-end system model is presented considering two users. Furthermore, outage probability calculations and theoretical error probability analysis are discussed and reinforced with computer simulation results.

	\end{abstract}

	\IEEEpeerreviewmaketitle
		\IEEEpubidadjcol
		
	\renewcommand{\thefootnote}{\fnsymbol{footnote}}
	\section{Introduction}
\IEEEPARstart{W}{ireless} communication systems are evolving rapidly in all communication layers, especially in the physical and MAC layers, to meet stringent 6G requirements. Specifically, in the physical layer, promising approaches and advancements have been investigated to enhance wireless communication quality and performance \cite{9144301,8732419,sahinofdm}. In general, the environment and channels in these wireless communication systems have been considered uncontrollable, which may lead to performance degradation by negatively affecting the communication efficiency, quality-of-service (QoS), and more \cite{basar2019wireless,zhang2021active}.

Recent advancements in meta materials have led to reconfigurable intelligent surfaces (RISs) to control the wireless channel passively and thus, enhance communication performance \cite{basar2019transmission,kilinc2021physical}. This popular technology has greatly attracted the attention of both academia and industry. An RIS is a large array with many cost-effective elements that can reflect, adjust, absorb and even amplify impinging signals with an adjustable reflecting coefficient \cite{9119122}. By doing so, RISs may artificially enhance existing wireless communication by suppressing interference \cite{OTAE,zhang2021spatial,interference_ref}, increasing capacity \cite{capacity_ref_1,capacity_ref_2}, extending coverage \cite{coverage_ref_1,coverage_ref_2,coverage_ref_3}, and providing energy efficiency \cite{fatih_recep_PA_paper,EE_ref,EE_ref_2}, all without the need for additional RF chains. RISs may also be used in applications such as localization, Wi-Fi, etc. Hence, the RIS technology is a promising tool to exploit in wireless communication systems. However, RISs do come with problems of their own. RISs typically work best in scenarios where direct links do not exist or are substantially weak \cite{zhang2021active, aygul2021deep, yang2020coverage}. Furthermore, a major challenge in RIS technology is the multiplicative path loss of the transmitter-RIS-receiver link, hence, making significant gains very difficult for passive RIS designs \cite{zhang2021active, wu2021intelligent, 9734027, fatih_recep_PA_paper}. 

\subsection{Related Work}

To combat the critical problem of multiplicative path loss, sacrificing from the fully passive structure, the active RIS concept has been introduced \cite{zhang2021active, 9734027}. In hopes of overcoming the high path loss problem, active RISs not only reflect the signals at the RIS but also amplify and strengthen the power of the impinging signals at the expense of additional power consumption and additional thermal noise introduced through the active elements \cite{9734027}. Active RIS technology brings additional flexibility to RIS, which is possible to achieve more than just mitigating the multiplicative path loss. By combining passive and active RIS technologies, hybrid systems may be constructed to enhance communication system performances and novel solutions \cite{wu2021intelligent,9591503}. Albeit its numerous benefits, active RIS designs also have drawbacks such as increased hardware complexity, increased power consumption and costs \cite{fatih_recep_PA_paper}. Hence, hybrid designs considering passive and active architectures are more promising. 

Passive and active RIS technologies have also been exploited and analyzed in non-orthogonal multiple access (NOMA) schemes. A simple design is proposed in \cite{ding2020simple} by employing multiple passive RISs and effectively aligning the direction of users' channel vectors to aid in NOMA. The sum rate is maximized in downlink MISO RIS-NOMA systems through joint beamforming designs at the BS and RIS. The downlink transmit power is minimized for RIS NOMA through joint optimization at the BS transmit beamformers and RIS phases. Moreover, the impact of coherent and random phase shift designs at the RIS are evaluated  \cite{mu2019exploiting,fu2019intelligent,ding2020impact}. In addition, a power domain downlink NOMA scheme is proposed by exploiting an RIS to improve the reliability of the two-user system and its bit-error rate is evaluated in \cite{thirumavalavan2020ber}. Furthermore, the impact of hardware impairments in terms of performance metrics is studied, outage probability and throughput performances are derived in \cite{9269324}. On the other hand, studies on uplink-NOMA exploiting RISs have also been conducted to enhance coverage, derive the outage probabilities and ergodic rates, and grant-free ultra-reliable uplink sensor data transmission exploiting diversity to obtain reliability \cite{cheng2021downlink,melgarejo2020reconfigurable}. Passive RIS-aided NOMA systems are prevalent in the literature; however, active and hybrid RIS-aided NOMA systems have yet to be explored, especially for uplink scenarios. As it can be concluded, RIS is a superior technology to aid NOMA wireless communication systems and is attracting the attention of researchers. However, the RIS-aided NOMA literature has yet to mature. Moreover, to the best of the authors' knowledge, an active or a hybrid RIS has not been the primary enabler of NOMA, especially in the uplink. In addition, RIS-assisted uplink-NOMA has not been considered in the literature.

\subsection{Main Contributions}
In this paper, a hybrid RIS is partitioned into two parts as active and passive parts. Benefiting from both active and passive RIS technologies, a novel system is proposed to enable uplink-NOMA over-the-air to the users, where the QoS of both users are enhanced by using a hybrid design rather than a fully active design. In conventional uplink-NOMA, the users require a sort of power allocation scheme to share the same time and frequency resources \cite{cheng2021downlink,9605575,9741239}. On the contrary, in the proposed method, the users may transmit with the same power and do not worry about interference with the confidence of the RIS  providing them their required QoS. The active part of the RIS is coherently aligned for one user and also amplifies its signal, on the other hand, the passive part is coherently aligned for the other user with no amplification in a fully passive manner. The first users signal is coherently aligned and boosted by the active part of the RIS to enable a certain power disparity from the other user, enabling NOMA and SIC at the BS. The passive part of the RIS requires no amplification since a power difference is already obtained by the active part, however, it is aligned for the second user to improve its signal-to-interference and noise ratio (SINR) as in conventional passive RIS methods. 

The major contributions of this paper are listed as follows:
\begin{itemize}
        \item A novel communication system is proposed to virtually enable uplink-NOMA over-the-air with the aid of a single hybrid RIS.
	    \item A thorough end-to-end system model of the proposed scheme is provided and analyzed along with computer simulations.
	    \item The outage probability is derived and calculated for different configurations and confirmed with simulations. 
	    \item A theoretical error probability analysis is provided along with computer simulation results for varying parameters and scenarios considering different successive interference cancellation (SIC) performances.
	    \item Two scenarios for a fixed active RIS gain case and an optimized RIS gain case are presented. The performances of both cases are compared and discussed.
	\end{itemize}

\subsection{Paper Organizations}
The rest of the paper is organized as follows, Section II presents the end-to-end system model of the proposed scheme. Section III presents a theoretical analysis of the outage probability for different scenarios. Section IV discusses the RIS power consumption and amplification factor. Furthermore, computer simulation results are provided and discussed in Section V . Finally, the paper is concluded with the final remarks in Section VI.

\section{End-To-End System Model}

\begin{figure*}[ht]
		\centering
		\includegraphics[width=0.8 \textwidth]{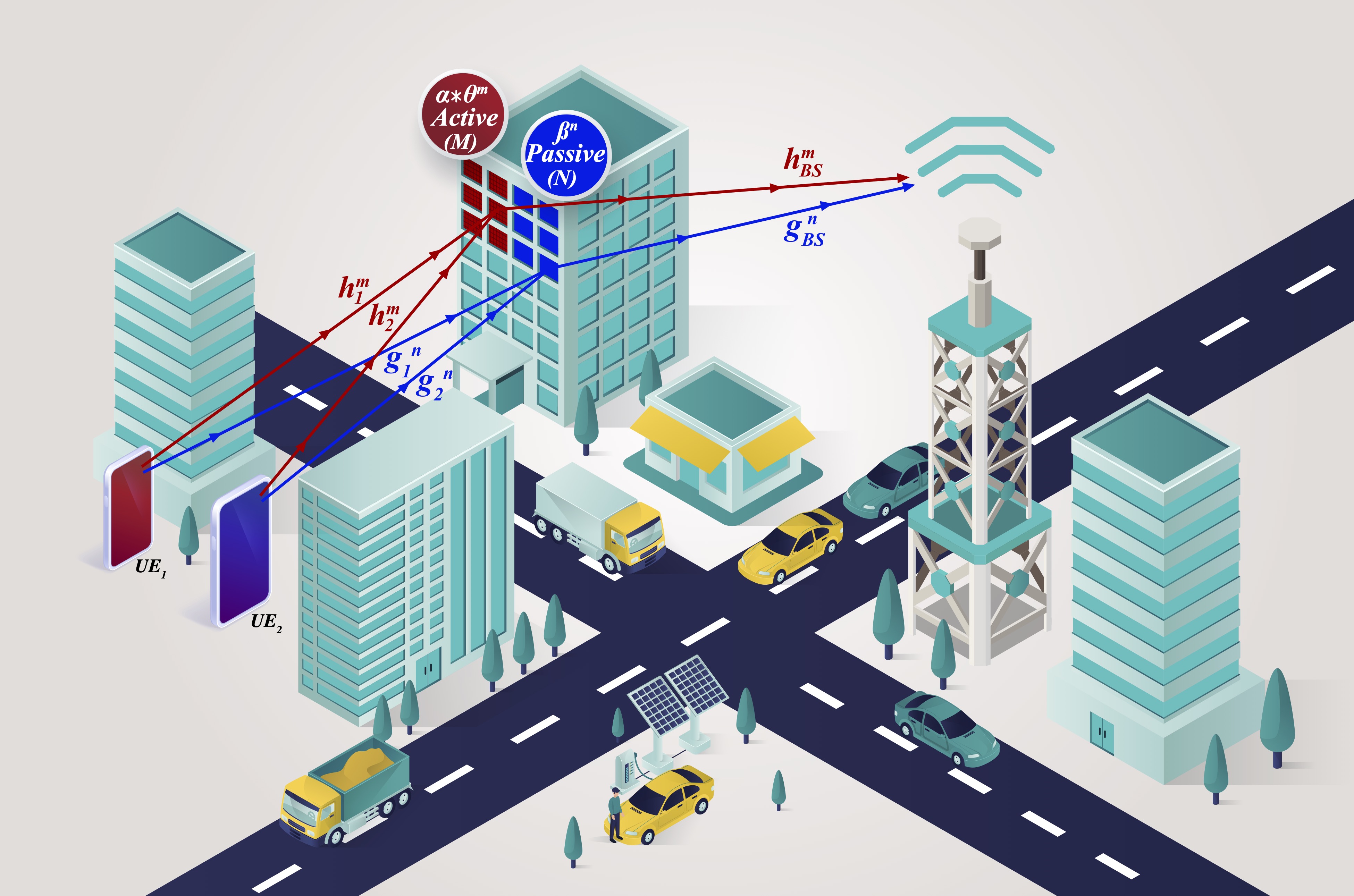} 
		\caption{ \centering Proposed over-the-air uplink-NOMA system model for $K=2$ users.}
		\label{fig:Fig1}
\end{figure*} 

This section presents the system model for the proposed over-the-air uplink-NOMA scheme. In the proposed system, we consider a two-user ($U_k, k\in{1,\cdots,K}$) network where $K=2$, uplink scenario operating under frequency-flat Rayleigh fading channels in the presence of an RIS. Both users are served through the same frequency and time slots with single-carrier signals. In conventional uplink-NOMA systems, the users' transmit power needs to be optimized to enable NOMA and different transmit powers are generally assigned to near and far users. In the proposed system, the users transmit with the same power $P_t$ without worrying about the power difference as power reception difference required by successive interference cancellation (SIC) receiver is satisfied by a hybrid RIS that is partitioned into two parts, i.e., active and passive as shown in Fig. \ref{fig:Fig1}. Assuming the CSI is available at the BS and RIS, the active part of the RIS with $M$ elements has amplifiers in every element and is coherently aligned for one user, either user 1 $(U_1)$ or user 2 $(U_2)$, while the other passive part with $N$ elements is coherently aligned for the other user. In other words, $\boldsymbol{\theta} \triangleq [e^{j\phi_1},\cdots,e^{j\phi_M}]^{\mathrm{T}} \in \mathbb{C}^{M\times 1}$ and $\boldsymbol{\beta} \triangleq [e^{j\xi_1},\cdots,e^{j\xi_N} ]^{\mathrm{T}} \in \mathbb{C}^{N\times 1}$ are the phase alignment vectors at the active and passive parts of the RIS, respectively, where $[\cdot]^{\mathrm{T}}$ represents transposition operator.

The received signal at the base station can be written as
\begingroup\makeatletter\def\f@size{9.5}\check@mathfonts
\begin{align}
    y = \sqrt{P_t}\sum_{k=1}^{K}&\left(\sqrt{\alpha}\sum_{m=1}^{M}h_k^m\theta^mh_{BS}^mx_k+\sum_{n=1}^{N}g_k^n\beta^n g_{BS}^n x_k\right) \nonumber \\ &+\sqrt{\alpha}\sum_{m=1}^{M}z^m\theta^mh_{BS}^m+w_0,
    \label{eq1}
\end{align}
\endgroup
where $x_k$, $\alpha$, $\theta^m$,  and $\beta^n$,  are the modulated PSK/QAM symbol transmitted by the $k^{th}$ user, amplification factor at the active part of the RIS, $m^{th}$ element of $\boldsymbol{\theta}$, $m={1,\cdots,M}$, and the $n^{th}$ element of $\boldsymbol{\beta}$,  $n={1,\cdots,N}$, respectively. Additionally, $z^m$ and $w_0$ stand for the noise term from the amplifier at the $m^{th}$ RIS element and additive white Gaussian noise (AWGN) sample which follows the distributions $\mathcal{CN}(0,\sigma_z^2)$ and $\mathcal{CN}(0,W_0)$, respectively. Furthermore, $h_k^m, g_k^n, h_{BS}^m$, and $g_{BS}^n$ respectively, stand for the independent and identically distributed  (i.i.d) Rayleigh fading channel coefficients of the $k^{th}$ user to the $m^{th}$ element of the active part of the RIS, $k^{th}$ user to the $n^{th}$ element of the passive part of the RIS, the $m^{th}$ element of the active part of the RIS to the BS, and the $n^{th}$ element of the passive part of the RIS to the BS.

The channel coefficients follow the distribution of $\mathcal{CN}(0,\sigma^2)$, where $\sigma^2=\frac{1}{L}$ and $L$ is the large scale path loss factor. For the first user ($U_1$), at the active part of the RIS, the $m^{th}$ element is coherently aligned as follows: $\theta^m=e^{j\phi^m}=e^{-j\angle(h_1^m h_{BS}^m)}$. On the other hand, for the second user ($U_2$), at the passive part of the RIS, the $n^{th}$ element is coherently aligned as follows: $\beta^n=e^{j\xi^n}=e^{-j\angle(g_2^ng_{BS}^n)}$. Equivalently, the received signal in (\ref{eq1}) may be rewritten as
\begin{align}
y &= \sqrt{P_t}\sum_{k=1}^{K}\left(\sqrt{\alpha}\mathbf{h}_{k}^{\mathrm{T}}\bar{\boldsymbol{\theta}} \mathbf{h}_{BS} +\mathbf{g}_{k}^{\mathrm{T}}\bar{\boldsymbol{\beta}}\mathbf{g}_{BS}\right)x_k \nonumber\\ &+\sqrt{\alpha}\mathbf{z}\bar{\boldsymbol{\theta}}\mathbf{h}_{BS}+w_0,
\label{eq2}
\end{align}
where $\mathbf{h}_{k}=[h_{k}^1,\cdots,h_{k}^M]^{\mathrm{T}}$, $\mathbf{g}_{k}=[g_{k}^1,\cdots,g_{k}^N]^{\mathrm{T}}$, $\mathbf{h}_{BS}=[h_{BS}^1,\cdots,h_{BS}^M]^{\mathrm{T}}$, and $\mathbf{g}_{BS}= [g_{BS}^1,\cdots,g_{BS}^N]^{\mathrm{T}}$ are the channel coefficient vectors between the $k^{th}$ user and the active and passive parts of the RIS, and between the RIS and BS for the active and passive parts, respectively. Furthermore, $\mathbf{z}=[z^1,\cdots,z^M]^{\mathrm{T}}$ is the noise term vector from tha active elements of the RIS. Additionally, $\bar{\boldsymbol{\theta}} = \mathrm{diag}(\boldsymbol{\theta})$ and $\bar{\boldsymbol{\beta}} = \mathrm{diag}(\boldsymbol{\beta})$, where $\mathrm{diag}(\cdot)$ is the diagonalization operator. Hence, the SINR of $U_1$ and $U_2$ may be expressed as in (\ref{eq3}), shown at the following page.
 



As in classical NOMA, $U_1$ directly decodes the its signal while $U_2$ requires SIC to decode $U_1$ and subtract its own signal from the received signal. It can be seen from Fig. \ref{fig:Fig1} as well as from the SINR equations, the signals of the users imping on all of the partitions of the RIS, thus, causing interference. However, as we will later discuss, this does not affect the NOMA communication due to the coherent aligning and amplification of the user making it much greater than the user supported by the passive side which is only coherently aligned and not amplified.

\begin{figure*}[t]
    \begin{equation}
        \begin{aligned}
           &\gamma_1  =\frac{P_t(|\sqrt{\alpha}\sum_{m=1}^{M}|h_1^m||h_{BS}^m|+\sum_{n=1}^{N}g_1^n\beta^n g_{BS}^n|^2)}{P_t(|\sqrt{\alpha}\sum_{m=1}^{M} h_2^m \theta^m h_{BS}^m +\sum_{n=1}^{N}|g_2^n| |g_{BS}^n||^2)+\sigma_z^2\alpha\sum_{m=1}^{M}|\theta^mh_{BS}^m|^2+W_0}
           \\
           &\gamma_2  =   \frac{P_t(|\sqrt{\alpha}\sum_{m=1}^{M} h_2^m \theta^m h_{BS}^m + \sum_{n=1}^{N}|g_2^n| |g_{BS}^n||^2)}{\sigma_z^2\alpha\sum_{m=1}^{M}| \theta^m h_{BS}^m|^2+W_0}
           \label{eq3}
        \end{aligned}
    \end{equation}
\vspace*{2pt}
    \hrule
\end{figure*}

\section{Theoretical Outage Probability Analysis}
 
In this section, the outage probability calculations are considered. For the sake of simplicity, the SINRs of the users impinging on the active and passive parts of the RIS are reexpressed as:
 \begin{align}
&\gamma_1 =  \cfrac{P_t|A+B|^2}{P_t|C+D|^2+\sigma_z^2\alpha\sum_{m=1}^{M}|E|^2+F}, \\
&\gamma_2 =  \cfrac{P_t|C+D|^2}{\sigma_z^2\alpha\sum_{m=1}^{M}|E|^2+F}.
\label{eq4n5}
\end{align}
 where, $A=\sqrt{\alpha}\sum_{m=1}^{M}|h_1^m||h_{BS}^m|$ , $B=\sum_{n=1}^{N}g_1^n\beta^n g_{BS}^n$, $C=\sqrt{\alpha}\sum_{m=1}^{M} h_2^m \theta^m h_{BS}^m$, $D=\sum_{n=1}^{N}|g_2^n| |g_{BS}^n|$, $E=\theta^mh_{BS}^m$, and $F=W_0$. Since $A$, $B$, $C$, and $D$ terms include summation over RIS elements and since the RIS elements are sufficiently large, due to the Central Limit theorem (CLT), the channel multiplications in $A$, $B$, $C$ and $D$ lead to Gaussian distribution regardless of their distributions. Due to the CLT, all channel multiplications lead to the Gaussian distribution, regardless of what the distribution of all parts $A$, $B$, $C$, $D$ are. These distributions are summarized in Table I. 
 
 \begin{lemma}
		\label{lemma:OPA}
		The mean and variance of $A$ are obtained as $\mu_A=\frac{\sqrt{\alpha}M\pi\sigma_{h_1}\sigma_{h_{BS}}}{4}$ and $\sigma_A=\alpha M\sigma_{h_1}^2\sigma_{h_{BS}}^2(1-\frac{\pi^2}{16})$, respectively. 
	\end{lemma}\begin{proof}
		We refer interested readers to Appendix A for the derivation steps.
	\end{proof} 
	
\begin{lemma}
		\label{lemma:OPB}
		The mean and variance of $B$ are obtained as $\mu_B=0$ and $\sigma_B=N\sigma_{g_1}^2\sigma_{g_{BS}}^2$, respectively. 
	\end{lemma}\begin{proof}
		We refer interested readers to Appendix B for the derivation steps.
	\end{proof} 
	
\begin{lemma}
		\label{lemma:OPC}
		The mean and variance of $C$, $D$, $E$, and $F$ are calculated as $\mu_C=0$, $\sigma_C=\alpha M\sigma_{h_2}^2\sigma_{h_{BS}}^2$, $\mu_D=\cfrac{N\pi \sigma_{g_2}\sigma_{g_{BS}}}{4}$, $\sigma_D=N\sigma_{g_2}^2\sigma_{g_{BS}}^2\left(1-\cfrac{\pi^2}{16}\right)$, $\mu_E=0$, $\sigma_E=\sigma_{BS}^2$, $\mu_F=0$, $\sigma_F=W_0$, respectively. 
	\end{lemma}\begin{proof}
		We refer interested readers to Appendix C for the derivation steps.
	\end{proof} 
 The outcomes of these three lemmas are summarized as in Table I.
\begin{table}[t!]
\caption{Random Variable (RV) Distributions For (3)} 
\centering 
\begin{tabular}{l c c r} 
\hline 
 RV  & $\mu$ & $\sigma^2$ & Type of RV
\\ [0.5ex]
\hline \hline 
\\[.3ex]
 $A$ & $\frac{\sqrt{\alpha}M\pi\sigma_{h_1}\sigma_{h_{BS}}}{4}$  & $\alpha M\sigma_{h_1}^2\sigma_{h_{BS}}^2(1-\frac{\pi^2}{16})$ & Real Gaussian   \\[2ex]
 \\[.3ex]
 $B$ & $0$ & $N\sigma_{g_1}^2\sigma_{g_{BS}}^2$ &  Complex Gaussian  \\[2ex]
 \\[.3ex]
 $C$ & $0$    & $\alpha M\sigma_{h_2}^2\sigma_{h_{BS}}^2$ &  Complex Gaussian   \\[2ex]
 \\[.3ex]
 $D$ & $\cfrac{N\pi \sigma_{g_2}\sigma_{g_{BS}}}{4}$    &  $N\sigma_{g_2}^2\sigma_{g_{BS}}^2\left(1-\cfrac{\pi^2}{16}\right)$ &  Real Gaussian \\[2ex]
 \\[.3ex]
 $E$ & $0$    & $\sigma_{BS}^2$ &  Complex Gaussian  \\[2ex]
 \\[.3ex]
 $F$ & $0$    & $W_0$ &  Complex Gaussian  \\[2ex]
 
 \hline

\hline 
\end{tabular}
\label{tab:PPer}
\end{table}

 
 In SINR expressions, squared sum of Gaussian random variables exist, which lead to the chi-squared distribution, which is a form of the Gamma distribution. Hence, the SINR is the ratio of chi-squared random variables, where $|A+B|^2$ and $|C+D|^2$ are non-central chi-squared random variables with two degrees of freedom and $\sigma_z^2\alpha\sum_{m=1}^{M}|E|^2$ is a central chi-squared random variable with $2M$ degrees of freedom. As seen in Fig. 2, the SINRs of both users fit the Gamma distribution.
 
 To obtain the outage probability ($OP$) of both users, we may use $OP_k=P(\gamma_k<2^{r_k}-1)$, where $r_k$ is the minimum rate at which the QoS holds for $Uk$. Since equivalent steps will be used for both users, for illustrative purposes, let us focus on $U_1$ supported by the active part of the RIS. The random variables in the SINR are left on one side and the CDF of the equation is solved for as
 
 \begin{figure}[t!]
		\centering
		\includegraphics[width=1.0\columnwidth,height=8cm]{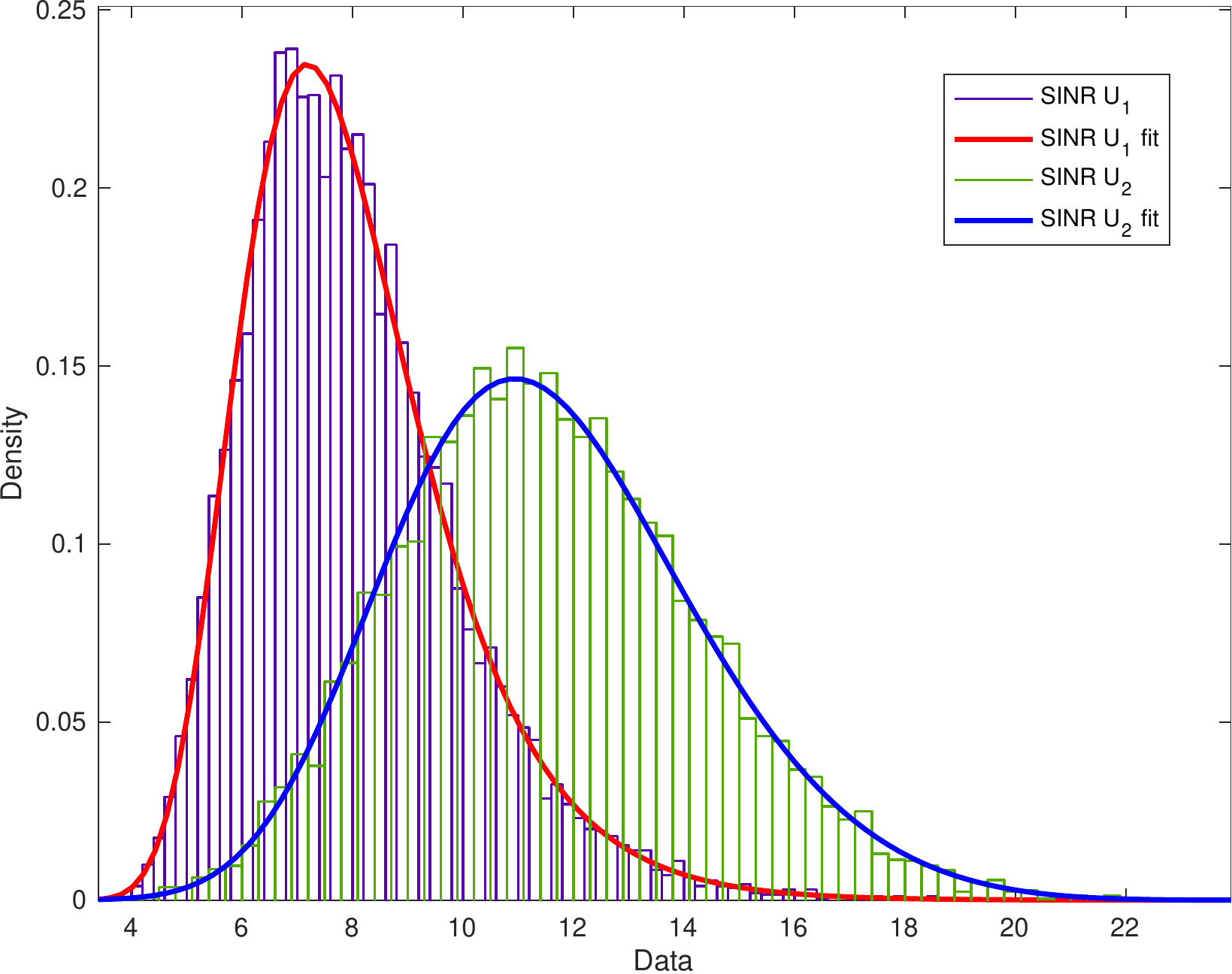}
		\caption{Fitted distribution of user SINRs.}
		\label{fig:Fig.2}
\end{figure}

 \begin{align}
&OP_k =  P\left(\cfrac{P_t|A+B|^2}{P_t|C+D|^2+\sigma_z^2\alpha\sum_{m=1}^{M}|E|^2+F}<v\right) \nonumber\\&= P(P_t|A+B|^2-P_t|C+D|^2v-\sigma_z^2\alpha\sum_{m=1}^{M}|E|^2v<Fv).
\label{eq6}
\end{align}

The term above may be simply expressed as the CDF of $G$, $P(G<g) = F_G(g)$, where $G=P_t|A+B|^2-P_t|C+D|^2v-\sigma_z^2\alpha\sum_{m=1}^{M}|E|^2v$ is a difference of multiple chi-square random variables, hence, complicating mathematical calculations. However, the characteristic functions (CFs) come in handy when evaluating chi-squared random variables. The CF of a chi-squared random variable is given as \cite{simon2002probability}

		\begin{equation}
\Psi_X(\omega)=\cfrac{1}{(1-2j\omega\sigma^2)^{n/2}}\exp{\left(\cfrac{j\omega\mu^2}{1-2j\omega\sigma^2}\right)},
\label{eq7}
		\end{equation} 
where $X=\sum_{k=1}^{n}X_k^2$ and $X_k \sim \mathcal{N} (\mu_k,\sigma^2).$ Hence, the CF of $E$ can be calculated. 

\begin{lemma}
		\label{lemma:OPD}
		Although the random variables $\{A,C\}$ and $\{B,D\}$ share the common terms $h_{BS}^m$ and $g_{BS}^n$, respectively, they are uncorrelated due to their zero correlation coefficient.
	\end{lemma}\begin{proof}
		We refer interested readers to Appendix D for the derivation steps.
	\end{proof} 

To conduct our theoretical analysis and obtain the characteristic functions, since the random variables are uncorrelated, we may express $|A+B|^2$ and $|C+D|^2$ in terms of their real and imaginary parts as $[\Re(A)+\Re(B)]^2+[\Im(B)]^2$ and $[\Re(C)+\Re(D)]^2+[{\Im}(C)]^2$, respectively. This is a sum of independent non-central chi-squared and central chi-squared random variables and since we know the statistics of the random variables the CFs may be calculated from \cite{simon2002probability} using
	\begingroup\makeatletter\def\f@size{8.5}\check@mathfonts
		\begin{align} 
\Psi_X(\omega)&=\cfrac{1}{((1-2j\omega\sigma_1^2)(1-2j\omega\sigma_2^2)^{n/2})}\exp{\left(\cfrac{j\omega\mu_1^2}{1-2j\omega\sigma_1^2}\right)}, \label{eq8}
\end{align}
\endgroup
where $\mu_1$, $\sigma_1$, and $\sigma_2$ will be substituted with the means and variances of the summation terms. Considering the sum of independent RVs as in \eqref{eq6}, is the multiplication of their individual CF's \cite{basar2020reconfigurable}, we obtain
\begin{align}
    &\Psi_G(w)= \nonumber \\ &\Psi_{P_t|A+B|^2}(w)\Psi_{(-P_t|C+D|^2v)}(w)\Psi_{(-\sigma_z^2\alpha\sum_{m=1}^{M}|E|^2v)}(w).\label{eq9}
\end{align}
Then, we may approach using the Gil-Pelaez's inversion formula \cite{quad_form}, since it is possible to obtain the cumulative distribution function (CDF) of this difference of chi-squared random variables as
 
 \begin{equation}
     F_G(g)=\cfrac{1}{2}-\int_{0}^{\infty}\cfrac{\Im\{e^{-jwg}\Psi_G(w)\}}{w\pi}dw, \label{eq10}
 \end{equation}
 where $\Im$ denotes the imaginary part, $F_G(g)$ is the CDF of $G$ and $\Psi_G(w)$ is the CF of $G$. With this CDF expression of the difference of chi-squares, it is possible to obtain the outage probability as mentioned previously.

\section{RIS Power Consumption and Amplification Factor}

The fully connected active RIS architecture includes a power amplifier (PA) in each RIS element and the power consumption of the active part of the RIS is dependent on the power consumed by the PAs. The power consumption of a single active RIS element is linear function of the output signal power of the PA and modeled as
\begin{align} 
    P_\textit{a}=\frac{1}{\nu} P_\textit{o}, \label{eq11}
\end{align}
where $\nu$ and $P_\textit{o}$ denote the efficiency and output power of the power amplifier, respectively. The total transmit power available to the active part of the hybrid RIS architecture is denoted by $P_t^\textit{RIS}$ and varied according to system requirements to sustain power domain NOMA. The output power of a single active RIS element is therefore obtained as $P_\textit{o}=P_t^\textit{RIS}/N$. To obtain the maximum output power from a single active RIS element, we need to maximize the gain of the amplifier by considering $P_\textit{o}$ constraint. Considering an ideal PA with $\nu=1$, the gain of the amplifier is obtained as follows
\begin{figure}[t!]
		\centering
		\includegraphics[width=1.0\columnwidth,height=8cm]{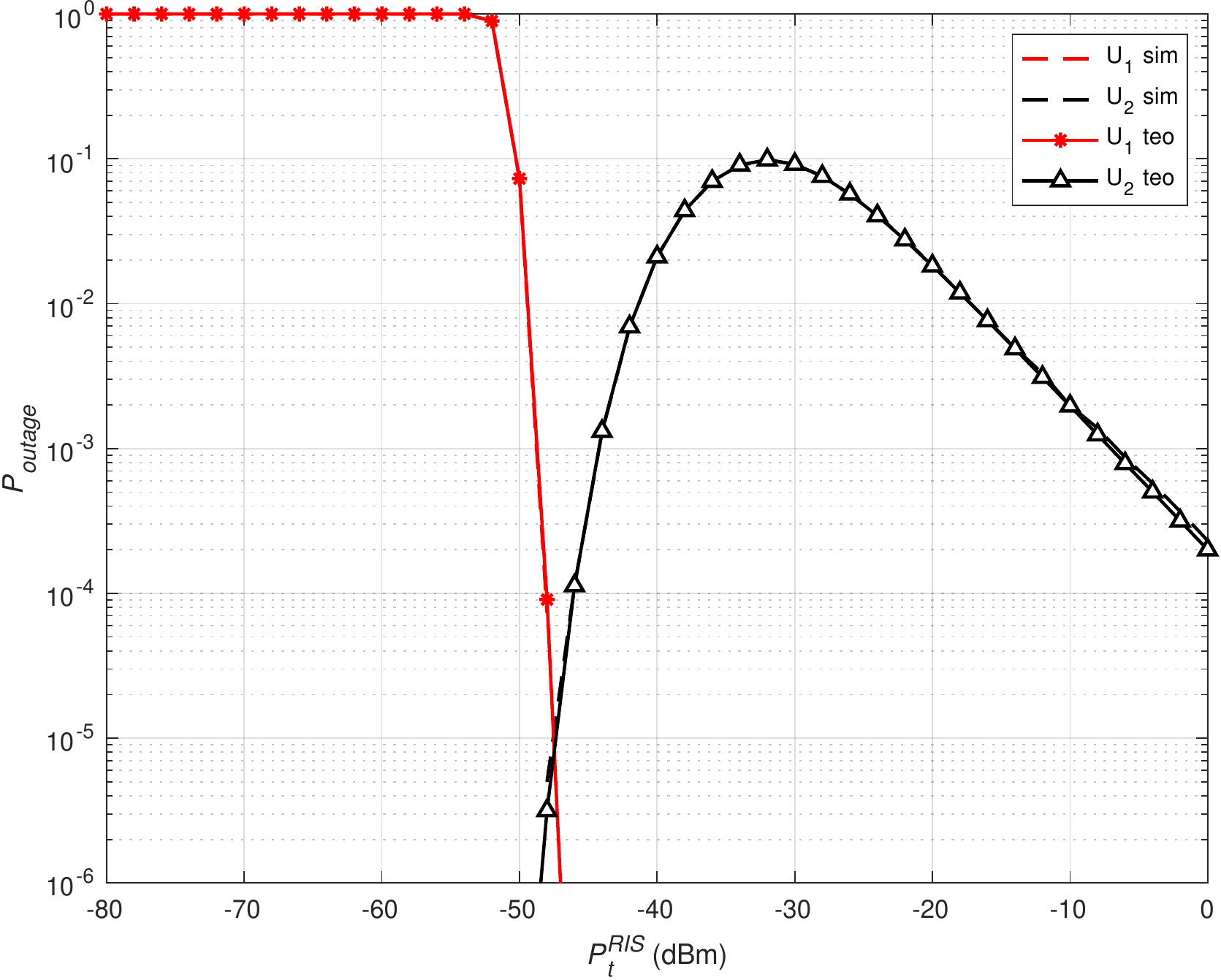}
		\caption{$P_t^\textit{RIS}$ vs outage probability to find optimal $\alpha$ for given default parameters.} 
		\label{fig:Fig.3}
\end{figure} 

\begin{align}
    G=\min\left(\sqrt{\frac{P_\textit{o}}{P_{t} \mathrm{E}[\|\mathbf{h}_{k}\|^2]}},G_\textit{max}\right) \label{eq12}
\end{align}
where $G_\textit{max}$ is the maximum gain of the amplifier and sets an upperbound on G. The denominator of the equation represents the average signal power at the input of the amplifier. Since we consider the uniform gain distribution for all active RIS elements, the average channel gain between the assisted user and active RIS elements is considered.

As seen in Eq. (\ref{eq12}), the amplifier gain $G$, also referred to as $\alpha$, depends on many variables and needs to be determined very carefully. The optimum gain should be obtained in order to best serve both users. Fig. \ref{fig:Fig.3} shows the effect of available power on the active RIS side for both users. It can been observed that for the parameters given in Fig. \ref{fig:Fig.3} which we consider as default parameters, the NOMA scheme can efficiently operate when -47 dBm amplifier power, which corresponds to $\alpha=8.5$ for each RIS element, is allocated for the active part. This result shows that the required power on the active RIS side is much less than the transmission power of the users and NOMA can be provided with a very small amount of power allocation and coherent alignment of the users.

Additionally, considering $\alpha$ is the same for all elements, the gain $\alpha$ can also be expressed in terms of the system parameters as $\alpha=\cfrac{P_t^\textit{RIS}/M}{P_tC_h}$, where $C_h$ is the average channel power of the users. 
It is important to note that for practicality, $\alpha$ is limited to a minimum of $0$ dB and a maximum of $30$ dB \cite{fatih_recep_PA_paper}. However, for the sake of better visualization explanations, Fig. 3 is not limited by the practical range of $\alpha$. The results mentioned in the remaining sections obey the aforementioned practical amplifier gain ranges. We optimize $P_t^{RIS}$ and then we extract the optimal $\alpha$ based on that optimum $P_t^{RIS}$ using the formulas above. $\alpha$ is selected in the fairest manner to serve both users. For the fixed $\alpha$ case, we look for the minimum difference of their outage performance which would be where they intersect as in Fig. \ref{fig:Fig.3}. However, for the fixed alpha case, the optimal alpha is best only for that specific set of parameters and as any parameter changes, it loses its optimality. 

\subsection{Problem Formulation}

In this case, we need to minimize $P_t^{RIS}$ while maximizing the outage probability performance which can be define by the optimization problem P1:

\begin{align}
        \mathrm{P}1=&\min_{P_t^{RIS},\Delta}\quad \Delta \nonumber \\
        &\hspace*{0.4cm}\text{s.t.}\hspace*{0.6cm} P_{out}^k \leq \Delta, \quad  k\in\{1,2\}. \label{eq13}
\end{align}
Here, $\Delta$ is the auxiliary variable representing the outage probability threshold constraint for a fair $P_t^{RIS}$ selection that serves both users. In P1, the aim is to obtain the best outage performance by using the minimum available power. To achieve this, we need to minimize both $P_t^{RIS}$ and $\Delta$. Since we aim at the fairness while obtaining the best performance with a minimum power consumption, we equivalently target the minimum difference in outage performance which provides the same result/concept. This can be formulated the and solved as the equivalent standard problem where P1 is transformed into P2:

\begin{align}
        \mathrm{P}2=&\min_{P_t^{RIS}} \quad \| P_{out}^1 - P_{out}^2 \| . \label{eq14}
\end{align}

Furthermore, we can solve P2 with a standard line search methods such as the Golden Selection algorithm and Simulated Annealing. However, different methods of $\alpha$ optimization may also be possible.
It should be noted that we ensure maximum fairness and all of the users have a weight no less than $\Delta$. 


\begin{table}[t!] 

	\centering
	\caption{Default Parameters}
	\label{tab:PLtable}
	{
	\begin{tabular}{lcc} 
	
		\hline 
		Parameter							& Value 		 \\ \hline \hline \\  [.3ex]
		$P_t^\textit{RIS}$ [dBm]				&  $-47$   \\  [2ex]
		 $\alpha$ (linear)		& $8.5$    \\ [2ex]
		 $M=N$		& $512$    \\ [2ex]
		 $P_t$ [dBm]  	& $15$  \\ [2ex]
		 $R^{th} $ [bps/Hz]  	& $2$  \\ [2ex]
		$\epsilon$ 		& $0$ (perfect SIC)  \\ [2ex]
		$W_0$ [dBm]		  	& $-130$  \\ [2ex]  
		$n_{amp}$ [dBm]	  	& $-130$  \\ [2ex]  
		\hline
	\end{tabular} 
	}
\end{table}

\section{Simulation Results}

In this section, the system of Fig. \ref{fig:Fig1} is revisited, where the transmission is carried out via an RIS under a blocked link between the users and the BS. For this setup, $d_{\text{$U_1$-BS}}$, $d_{\text{$U_2$-BS}}$, $d_{\text{$U_1$-RIS}}$, $d_{\text{$U_2$-RIS}}$, $d_{\text{RIS-BS}}$, represent the distances between the users and BS, users and RIS, RIS and BS, respectively. In addition,  $h_{\text{$U_1$}}$, $h_{\text{$U_2$}}$, $h_{\text{RIS}}$, $h_{\text{BS}}$, and $f_c$, represent the height of the users, RIS, BS, and the operating frequency, respectively. The 3GPP UMi point-to-point NLOS path loss model for a single Tx/Rx path is also considered \cite{3rdGP}. Within the $2-6$ GHz frequency band and Tx-Rx distance ranging from $10-2000$ m, the 3GPP UMi path loss model for NLOS transmission is expressed as:
\begin{equation}
    L(d)[\text{dB}]=36.7\text{log}_{10}(d)+22.7+26\text{log}_{10}(f_c). \label{eq15}
\end{equation}
  The RIS is composed of $M$ active and $N$ passive and controllable reflecting elements. 
  
For all computer simulations, path loss is included and the default parameters in Table II as well as the setup in Table III are used as the fixed parameters for simulations and analyses unless stated otherwise. Fixed $\alpha=8.5$ is considered through out the paper, as it is the optimal case for the default parameters as seen from Fig. \ref{fig:Fig.3}, unless it is specifically stated that the optimized case of $\alpha$ is use.

\begin{table}[t!] 

	\centering
	\caption{Computer Simulation Set-Up Parameters}
	\label{tab:PLtable}
	{
	\begin{tabular}{lcc} 
	
		\hline
		Parameter							& Value 		 \\ \hline \hline \\  [.3ex]
		$f_c$ [GHz]				&  $5$  \\  [2ex]
		 $d_{\text{$U_1$-BS}} $[m]		& $55.73$    \\ [2ex]
		 $d_{\text{$U_2$-BS}} $[m]		& $55.73$    \\ [2ex]
		 $d_{\text{$U_1$-RIS}} $[m]   	& $35.51$  \\ [2ex]
		 $d_{\text{$U_2$-RIS}} $[m]   	& $35.51$  \\ [2ex]
		$d_{\text{RIS-BS}}$ [m]		& $20.22$   \\ [2ex]
		$h_{\text{$U_1$}}$ [m]		  	& $10$  \\  [2ex]
		$h_{\text{$U_2$}}$ [m]		  	& $10$  \\  [2ex]
	    $h_{\text{RIS}}$ [m]		& $4$    \\ 	[2ex]
		$h_{\text{BS}}$ [m]      	& $1$  \\ [2ex]
		$W_0$ [dBm]	        	    & $-130$    \\ [2ex]
		\hline
	\end{tabular} 
	}
\end{table}

\begin{figure}[t!]
		\centering
		\includegraphics[width=1.0\columnwidth,height=8cm]{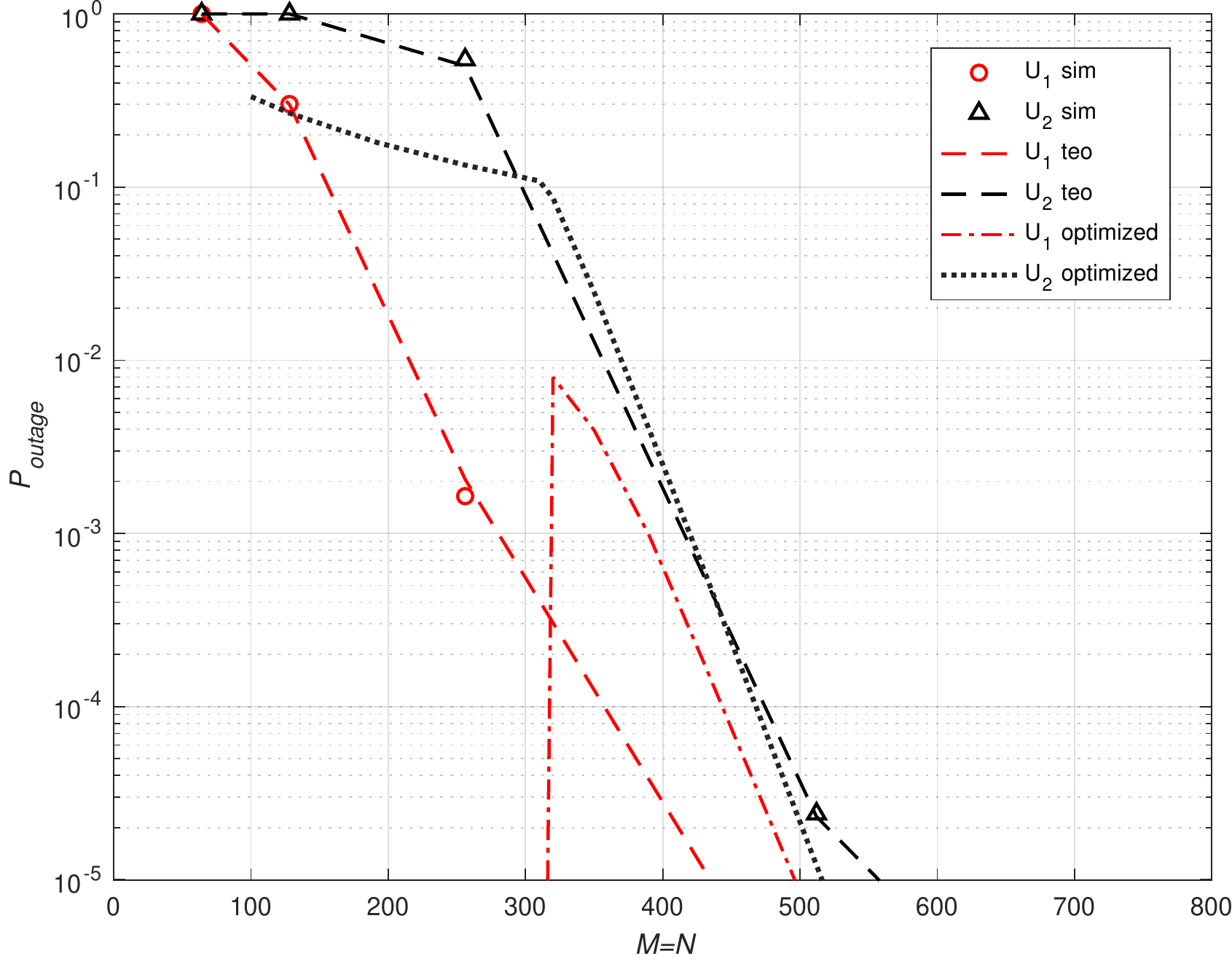}
		\caption{Outage probability performance with respect to varying RIS size.}
		\label{fig:Fig.4}
\end{figure} 

Typically in an RIS-aided communication system, the performance increases with the size of RIS. This behaviour can also be seen in Fig. \ref{fig:Fig.4}, which includes the outage probability performance simulation and theoretical simulation for both users with respect to varying RIS size. It can be observed that for a fixed $\alpha$, a minimum RIS size of $300$ is needed to enable communication at both users, however, a larger RIS is preferred for a enhanced performance at the cost of additional hardware expenditure. When $\alpha$ is optimized, it can be seen that the performance for both users are enhanced and a smaller RIS size is sufficient to provide the same performance compared to the fixed $\alpha$ case. $U_1$ has perfect communication while $U_2$ is at outage until an RIS size of $M=N=300$. Here, the optimization notices that $U_2$'s performance is insufficient and identifies no purpose in arranging $\alpha$ to aid it and arranges $\alpha$ completely for $U_1$. This is why $U_1$'s performance is very high. Once $U_2$'s performance starts to reach an acceptable performance, the optimization switches on and arranges $\alpha$ for both users to enable NOMA. This is why there is a sharp spike in $U_1$'s performance. Hence, $U_1$'s performance degrades to a certain point. Then as the RIS size increases, both users' performance increase to very good values while NOMA continues to be executed by the RIS.

\begin{figure}[t!]
		\centering
		\includegraphics[width=1.0\columnwidth,height=8cm]{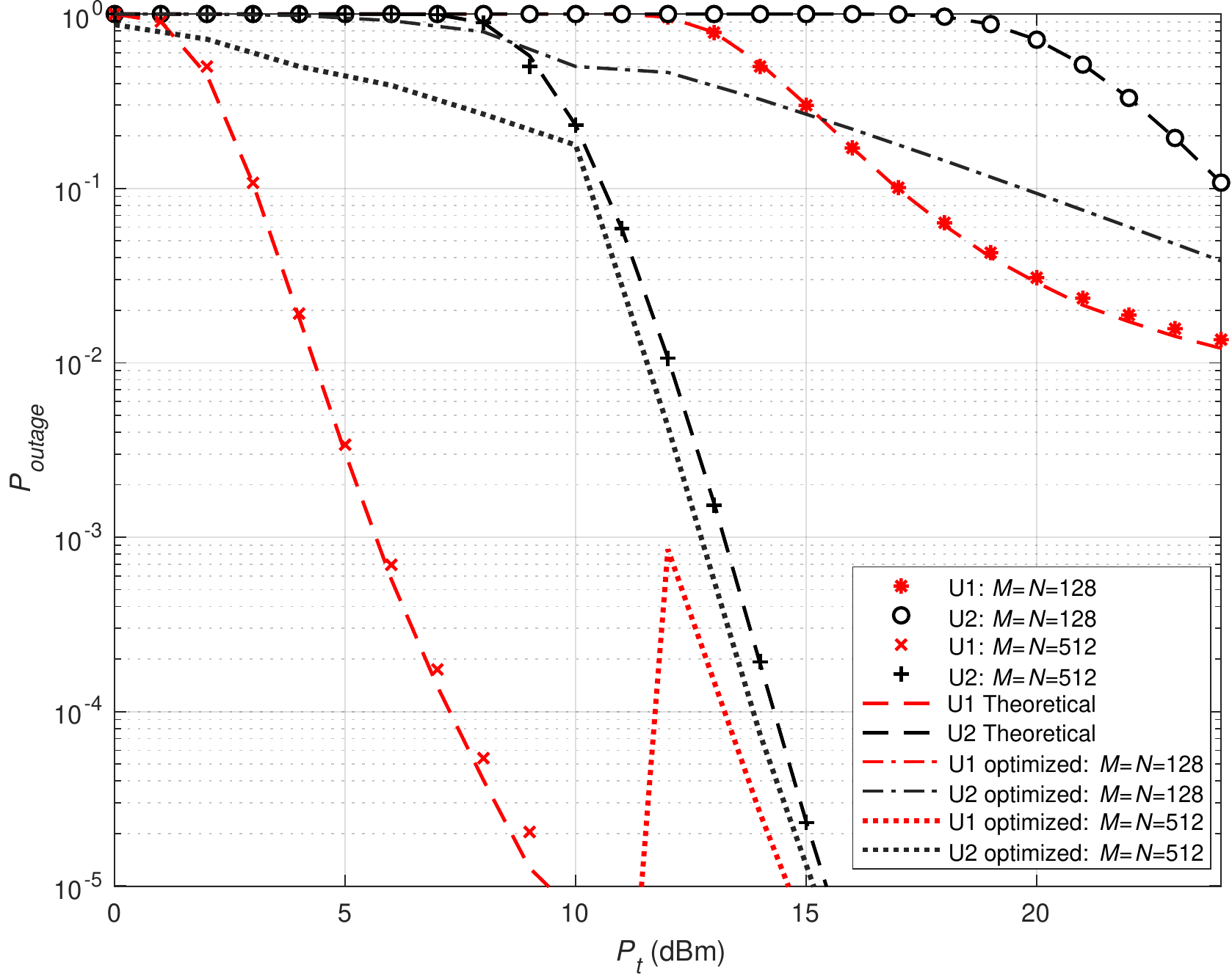}
		\caption{Outage probability performance as users transmit powers $P_t$ vary.} 
		\label{fig:Fig.5}
\end{figure} 

Furthermore, Fig. \ref{fig:Fig.5} presents the outage probability of both users for a varying transmit power, $P_t$ (dBm), of the users. It is observed that for both RIS sizes, an increase in $P_t$ enhances the performance for both users. However, it should be noted that the maximum transmit power at the users for uplink is $23$ (dBm) \cite{haider2019maximum}. It can be seen once more that as the RIS size increases, the performance drastically increases, especially for a high $P_t$. Additionally, for a smaller RIS size and high transmit powers, the performance reaches to a saturation and has an error floor. Hence, increasing the transmit power would not increase the performance after a certain point for a given RIS size. When comparing the fixed and optimized $\alpha$ scenarios, it is clear that the optimized case significantly outperforms the fixed $\alpha$ case. $U_1$ for $M=N=128$ performance is very good below $10^{-5}$, since $U_2$'s performance is not great for a low RIS size and since uplink transmit power below 23 (dBm) is not enough to boost $U_2$. However for $M=N=512$ we see that $U_2$ is in outage until around $9$ dBm and hence, $\alpha$ is optimized only for $U_1$ enhancing its performance and reducing the outage probability. Once $U_2$ reaches $10^{-1}$, $\alpha$ starts optimizing itself for both users and it can be seen from the spike of degradation in $U_1$'s performance. However, as $P_t$ increases to $14$ dBm, both users provide excellent outage probability performance.

\begin{figure}[t!]
		\centering
		\includegraphics[width=1.0\columnwidth,height=8cm]{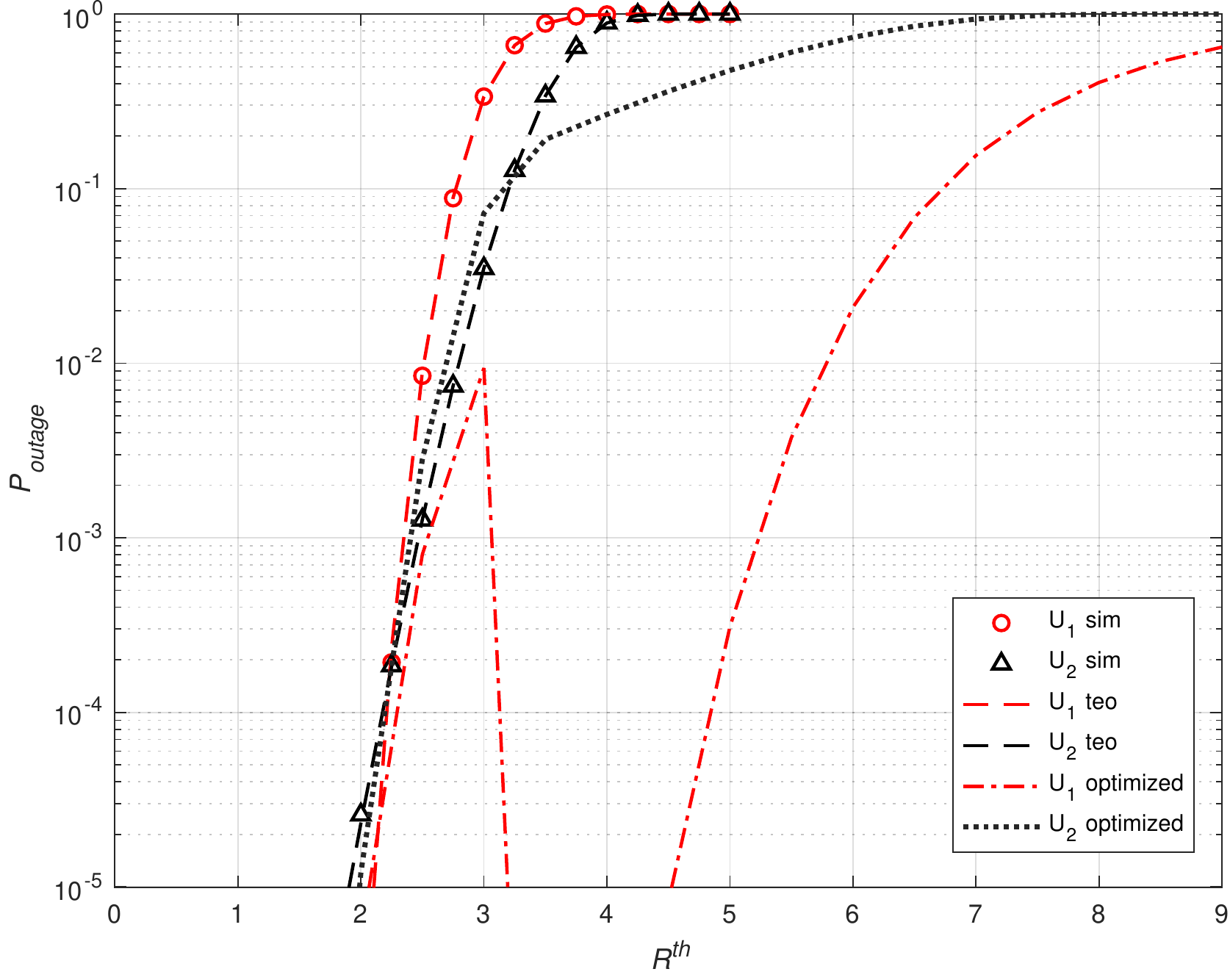}
		\caption{Outage probability performance of users as the QoS requirement is increased.} 
		\label{fig:Fig.6}
\end{figure} 

As the QoS requirement for a user's communication performance gets higher, the outage probability increases since it is harder to sustain that quality. This can be seen in Fig. \ref{fig:Fig.6} for both users presenting the fixed and optimized case. For the fixed user case, it can be seen that for both users, QoS of $R^{th}=2$ bps/Hz is easily sustained for the default parameters and as the quality requirement increases the outage probability increases and by $4$ bps/Hz both users are at outage. It should be noted that this is for the default parameters and the performance can be arranged to meet certain requirements by changing these parameters. However, with the same default parameters and only difference of using an optimized alpha, the QoS is enhanced drastically. Both users easily sustain the quality requirement from $0-2$ bps/Hz and the performance degrades as the QoS is increased to $3$ bps/Hz. Since $U_2$'s outage performance get very low at $3$ bps/Hz, the $\alpha$ optimization disregards $U_2$ and arranges itself for $U_1$ thus, once again providing it outstanding performance below $10^{-5}$ which is shown by the spike down at $R^{th}=3$. As the $R^{th}$ increases to $9$ bps/Hz, $U_1$'s performance slowly degrades once again. However, it is clearly seen that when $\alpha$ is optimized, $U_2$ reaches complete outage at $7$ bps/Hz rather than $4$ bps/Hz at fixed $\alpha$ and $U_1$ reaches complete outage at approximately $11$ bps/Hz rather than $3.3$ bps/Hz at the fixed $\alpha$ case.

\begin{figure}[t!]
		\centering
		\includegraphics[width=1.0\columnwidth,height=8cm]{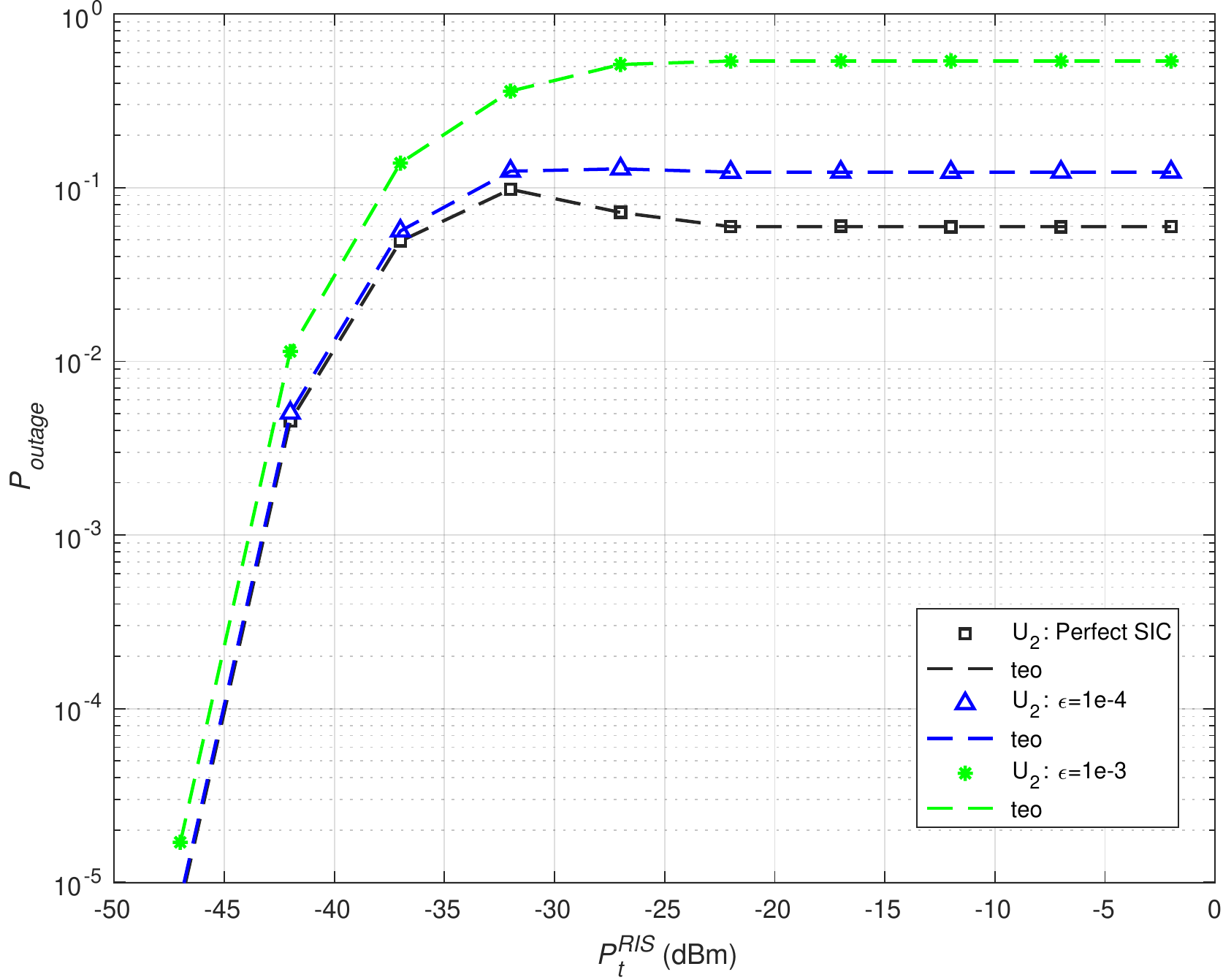}
		\caption{Affect of RIS transmit power with SIC imperfections to outage probability.}
		\label{fig:Fig.7}
\end{figure} 

\begin{figure}[t!]
		\centering
		\includegraphics[width=1.0\columnwidth,height=8cm]{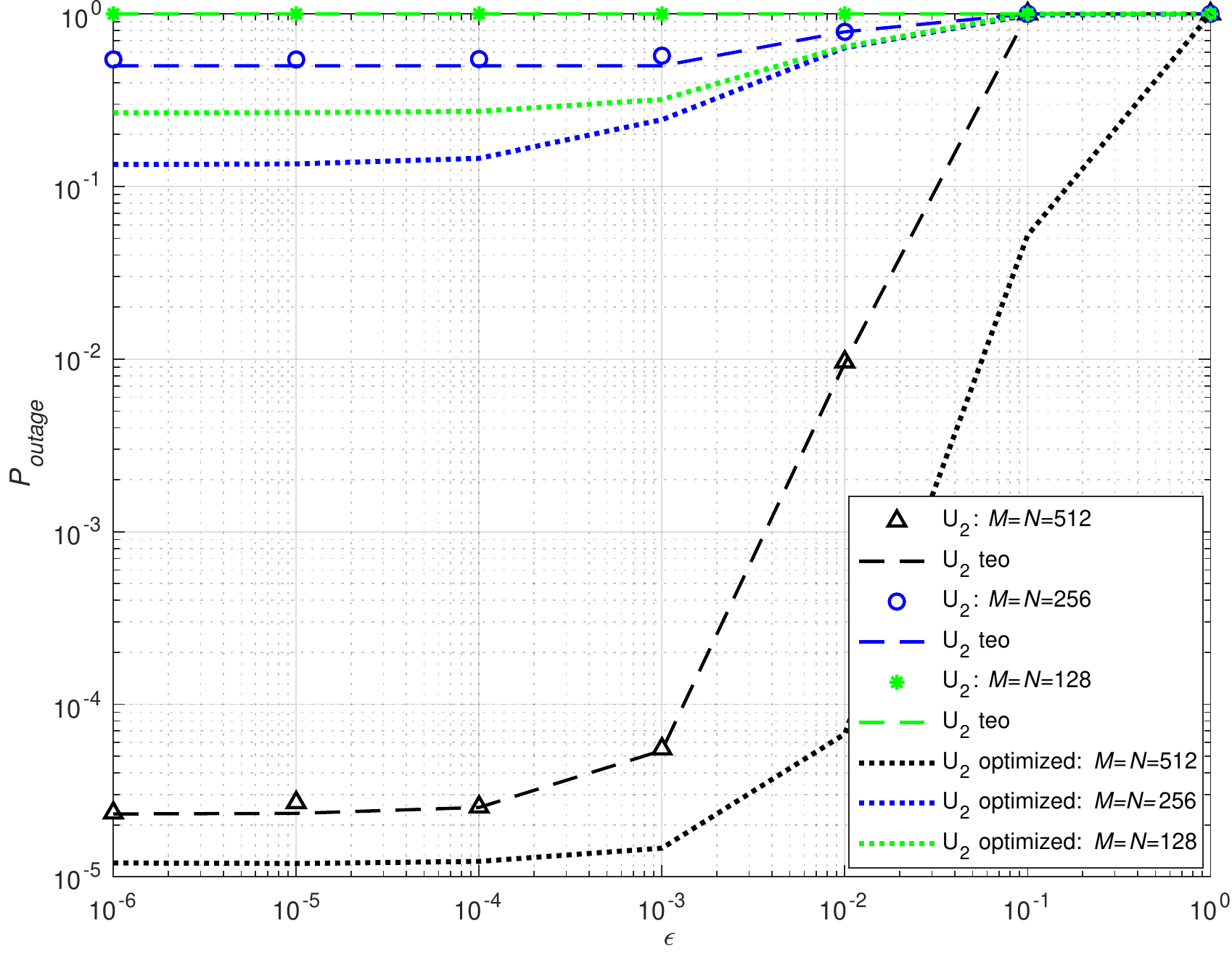}
		\caption{Effect of $\epsilon$ to outage probability}
		\label{fig:Fig.8}
\end{figure}

Fig. \ref{fig:Fig.7} and \ref{fig:Fig.8} present the cases of perfect and imperfect SIC when $U_2$ decodes $U_1$'s signal to subtract it from its received signal. Since $U_1$ is not affected by errors in SIC decoding, it is not included in these figures. In Fig. \ref{fig:Fig.7}, we can see the outage probabilities for $U_2$ at default parameters with different SIC performances ($\epsilon$ values) while the transmit power of the RIS $P_t^\textit{RIS}$ varies. It is shown that perfect SIC ($\epsilon$=0) provides the best performance whereas cases of imperfect SIC results in lower performance. Here, $P_t^\textit{RIS}$ is very low even though an active RIS is used. It should be noted that as the $P_t^\textit{RIS}$ value increases even further, the outage probability would approach to $0$. However, since $\alpha$ is limited to a practical value of $1000$, as we increase the $P_t^\textit{RIS}$, $U_2$'s performance does not improve after getting worse. At low transmit power, $U_2$ experiences great performance since $\alpha$ value is low and $U_2$ is aligned coherently at the passive section which strengthens the performance from the passive part making the active part of the hybrid RIS negligible. However, as $P_t^\textit{RIS}$ increases, the $\alpha$ value increases thus, the gap between the strength of the passive and active parts of the RIS decreases. This decrease in difference of strength from both sides degrades the performance because as the active part of the RIS increases, the random phase alignment strength increases and corrupts the perfectly aligned coherent phase part of the passive part for $U_2$. As the $P_t^\textit{RIS}$ increases the active part will dominate the passive part and the outage probability performance will get better. However, since we limited $\alpha$ to a practical value, we see a saturation in the performance after $-25$ dbm. 

On the other hand, Fig. \ref{fig:Fig.8} presents the effect epsilon has to the outage probability performance of $U_2$ for fixed $P_t^\textit{RIS}$ and with different RIS sizes. As $\epsilon$ is closer to zero, the performance is the highest and get worse as $\epsilon$ approaches higher values meaning higher SIC errors. Once again, higher RIS sizes provides superior performance and the optimized alpha cases outperform the fixed $\alpha$ scenarios. 


\section{Conclusion}		

The main motivation of this paper is to propose a novel method to achieve power domain NOMA in the uplink, through the use of the hybrid RIS. Conventionally, while the power disparity of the users are implemented at either the users or the base-station, the power difference is performed over-the-air at the RIS. A thorough end-to-end system model along with theoretical calculations are provided and analyzed with various parameters. Future and open research areas include varying the size of the hybrid RIS partitions where they are not equal, further optimization of the alpha parameter, optimal placement of the RIS and other varying system parameters, and the possibility of considering other domains to enhance the NOMA performance.

	\balance

\newpage

	\appendices
	\numberwithin{equation}{section}
	\section{Proof of Lemma \ref{lemma:OPA}}
	\label{AppOP}
Let us define 
\begin{align}
    A=\sqrt{\alpha}\sum_{m=1}^{M}|h_1^m||h_{BS}^m|.
\end{align}
The active part of the RIS aligns coherently to cancel out the phases of the complex Gaussian random variables $h_1^m$ and $h_{BS}^m$ hence, resulting in only the magnitude of the channels $|h_1^m|$ and $|h_{BS}^m|$ which are i.i.d. Rayleigh distributed. Since they are i.i.d. random variables, the mean can be expressed as the product of their individual means as 
 \begin{align}
     \mathrm{E}[|h_1^m||h_{BS}^m|]=\mathrm{E}[|h_1^m|]\mathrm{E}[|h_{BS}^m|].
 \end{align}
 Thus (A.2) equates to
 \begin{align}
     \left(\sigma_{h_1}\sqrt{\frac{\pi}{4}}\right)  \left(\sigma_{h_{BS}}\sqrt{\frac{\pi}{4}}\right) =\frac{\sigma_{h_1}\sigma_{h_{BS}}\pi}{4}.
 \end{align}
 
On the other hand, the variance is calculated using 
 
 \begin{align}
     \mathrm{VAR}[|h_1^m||h_{BS}^m|]=\mathrm{E}[(|h_1^m||h_{BS}^m|)^2]-\mathrm{E}[|h_1^m||h_{BS}^m|]^2.
 \end{align}
 
The calculation of $\mathrm{E}[|h_1^m||h_{BS}^m|]^2$ is simple as we just square (A.3) resulting in
\begin{align}
    \mathrm{E}[|h_1^m||h_{BS}^m|]^2=\cfrac{\sigma_{h_1}^2\sigma_{h_{BS}}^2 \pi^2}{16}.
\end{align}

To obtain $\mathrm{E}[(|h_1^m||h_{BS}^m|)^2]$, which can be re expressed as $\mathrm{E}[|h_1^m|^2]\mathrm{E}[|h_{BS}^m|^2]$, we need to find the individual second moments.

We will find the second moment for $|h_1^m|$ as $\mathrm{E}[|h_1^m|^2]$ since the same can be applied for the latter. The variance will help us to do so:

\begin{align}
    \mathrm{VAR}[|h_1^m|]=\mathrm{E}[|h_1^m|^2]-\mathrm{E}[|h_1^m|]^2.
\end{align}
Reordering allows us to solve for the desired $\mathrm{E}[|h_1^m|^2] $ as
\begin{align}
    \mathrm{E}[|h_1^m|^2]  &=VAR[|h_1^m|]+\mathrm{E}[|h_1^m|^2] \\ 
    &=\cfrac{4\sigma_{h_1}^2}{4}-\cfrac{\pi}{2}\cdot\cfrac{\sigma_{h_1}^2}{2}+\cfrac{\sigma_{h_1}^2\pi}{4} \\
    &=\sigma_{h_1}^2
\end{align}
The same is applied to $\mathrm{E}[|h_{BS}^m|^2]$ resulting in $\sigma_{h_{BS}}^2$.
Thus,  $\mathrm{E}[(|h_1^m||h_{BS}^m|)^2]$ is expressed as 
 \begin{align}
     \mathrm{E}[(|h_1^m||h_{BS}^m|)^2]=\sigma_{h_1}^2\sigma_{h_{BS}}^2.
 \end{align} 
Finally, to obtain the variance of $|h_1^m||h_{BS}^m|$ we substitute (A.5) and (A.10) into (A.4) to obtain
 \begin{align}
    \mathrm{VAR}[|h_1^m||h_{BS}^m|]= \sigma_{h_1}^2\sigma_{h_{BS}}^2(1-\frac{\pi^2}{16}).
 \end{align}
 
 When the number of RIS elements $M$ is significantly large, from the CLT, the random variable $A$ from (A.1) converges to a real Gaussian distributed random variable with $A \sim \mathcal{N}(\frac{\sqrt{\alpha}M\pi\sigma_{h_1}\sigma_{BS}}{4},\,\alpha M\sigma_{h_1}^2\sigma_{BS}^2(1-\frac{\pi^2}{16}))\,$.

 \section{Proof of Lemma \ref{lemma:OPB}}
	\label{AppOP}
 
 For $B$, we have the complex Gaussian random variable as 
 \begin{align}
     B=\sum_{n=1}^{N}g_1^n\beta^n g_{BS}^n.
 \end{align}
 
 The passive part of the RIS $(\beta^n)$ does not align coherently for the channels, hence, the phases do not cancel out and remains a complex Gaussian distributed random variable. Thus the mean remains zero and the variance is the multiplication of the variances of the two i.i.d random variables $g_1^n$ and $g_{BS}^n$ as
 
 \begin{align}
     \mathrm{VAR}[g_1^ng_{BS}^n] = \mathrm{VAR}[g_1^n]\mathrm{VAR}[g_{BS}^n]= (\sigma_{g_1}^2)(\sigma_{g_{BS}}^2).
 \end{align}
 Due to sufficiently large $N$ elements from the RIS, and the CLT, the variance can be finalized as 
 \begin{align}
     \mathrm{VAR}[B]=N\sigma_{g_1}^2\sigma_{g_{BS}}^2.
 \end{align}

 \section{Proof of Lemma \ref{lemma:OPC}}
	\label{AppOP}
	
 Furthermore, the means and variances or the random variables $C$ and $D$ are found similarly as in $A$ and $B$. $C$ is not coherently aligned with the active part of the RIS, therefore the approach in $B$ is exploited. $D$ is coherently aligned with the passive part of the RIS, hence, the identical approach of $A$ is considered.
 
$E$ is simply the multiplication of the phase shift of the RIS element $e^{j\phi_m}$ and $h_{BS}^m$. Since the mean and variance do not change when multiplied with an exponential term, $E$ follows the same distribution as $h_{BS}^m$ which is $\mathcal{CN}(0,\sigma_{BS}^2)$. $F$ is simply the noise term and its distribution has been provided.
  
  \section{Proof of Lemma \ref{lemma:OPD}}
	\label{AppOPD}
 
As it can be seen from (3), $A$ and the $C$ have $h_{BS}^m$ as a common term; however, they are independent since they are uncorrelated and Gaussian distributed. To find their correlation coefficient, the covariance of $A$ and $C$ is required and since $C$ has zero mean, it may be expressed as
\begin{equation}
\mathrm{E}[AC]=\mathrm{E}\left[\left(\sum_{m=1}^{M}h_1^me^{j\theta_m}h_{BS}^m\right)\left(\sum_{m=1}^{M}h_{2}^me^{j\theta_m}h_{BS}^m \right)\right].
\end{equation}
After adjusting the RIS phases according to the first path, (D.1) may be re-expressed as
\begin{equation}
\mathrm{E}[AC]=\mathrm{E}\left[\left(\sum_{m=1}^{M}|h_1^m||h_{BS}^m|\right)\left(\sum_{m=1}^{M}h_{2}^me^{-j\angle{h_2^m}}|h_{BS}^m|\right)\right].
\end{equation}
Here, $h_{1}$ is a complex Gaussian term with zero mean hence, $\mathrm{E}[AC]=0$. Therefore, the correlation coefficient of $A$ and $C$ is also zero resulting in $A$ and $C$ being uncorrelated. Furthermore, they are Gaussian distributed and independent. On the other hand, the same procedure applies for $B$ and $D$ where they share the $g_{BS}^n$ term.

\bibliographystyle{IEEEtran}
\bibliography{IEEEabrv,references}
\end{document}